\RequirePackage[2020-02-02]{latexrelease}
\documentclass[onecolumn,amsmath,amssymb,10pt,aps]{revtex4}
  
\usepackage[utf8]{inputenc}
\usepackage[T1]{fontenc}

\usepackage{amsmath}
\usepackage{amsfonts}
\usepackage{graphicx}
\usepackage{yfonts}
\usepackage{color}
\usepackage[normalem]{ulem}
\usepackage{amsthm}
\usepackage{bm}
\usepackage{bbm}
\usepackage{mathtools}
\usepackage{array}
\usepackage{placeins}
\usepackage{enumitem}

\newcommand{\der}{\mathrm{d}}
\newcommand\bigforall{\mbox{\Large $\mathsurround=1pt\forall$}}

\def\<{\langle}
\def\>{\rangle}

\newcommand{\Tr}{\mathrm{Tr}}
\def\oper{{\mathchoice{\rm 1\mskip-4mu l}{\rm 1\mskip-4mu l}
{\rm 1\mskip-4.5mu l}{\rm 1\mskip-5mu l}}}
\DeclareMathAlphabet\mathbfcal{OMS}{cmsy}{b}{n}

\newtheorem{Theorem}{Theorem}
\newtheorem{Lemma}{Lemma}

\newtheorem{Remark}{Remark}
\newtheorem{Proposition}{Proposition}

\begin{document}

\title{Non-Markovian quantum dynamics from symmetric measurements}

\author{Katarzyna Siudzi{\'n}ska}
\affiliation{Institute of Physics, Faculty of Physics, Astronomy and Informatics \\  Nicolaus Copernicus University, Grudzi\k{a}dzka 5/7, 87--100 Toru{\'n}, Poland}

\begin{abstract}
We use symmetric measurement operators to construct quantum channels that provide a further generalization of generalized Pauli channels. The resulting maps are bistochastic but in general no longer mixed unitary. We analyze their important properties, such as complete positivity and the ability to break quantum entanglement. In the main part, we consider the corresponding open quantum systems dynamics with time-local generators. From divisibility properties of dynamical maps, we derive sufficient Markovianity and non-Markovianity conditions. As instructive examples, we present the generators of P-divisible generalized Pauli dynamical maps that allow for more negativity in the decoherence rates.
\end{abstract}

\flushbottom

\maketitle

\thispagestyle{empty}

\section{Introduction}

Non-Markovian quantum evolution finds important applications in quantum information processing and
quantum communication, where it becomes crucial to include the effects of memory \cite{BLPV,Alonso}. With a rapid development of experimental methods, it becomes possible to detect \cite{Bernardes,WuHou,Walborn}, quantify \cite{Balthazar}, and control \cite{LiuHuang,YuWang} even weak effects caused by the system-environment interactions. Open quantum systems with dissipative dynamics are simulated using Noisy Intermediate Scale Quantum (NISQ) devices \cite{Huggins,Monroe,Preskill}. Examples include decoherence channels on superconducting quantum computers \cite{GarciaPerez}, photosynthetic systems with dephasing noises \cite{WangNISQ}, photosynthetic light harvesting \cite{Croce}, and electron transfer in organic materials \cite{Rafiq}.

Markovianity of open quantum system evolutions is encoded in the properties of dynamical maps $\{\Lambda(t):\,t\geq 0,\,\Lambda(0)=\oper\}$, where $\Lambda(t)$ are completely positive, trace-preserving maps (quantum channels). It is said that $\Lambda(t)$ is divisible if there exists a trace-preserving propagator $V(t,s)$ such that $\Lambda(t)=V(t,s)\Lambda(s)$ at any times $t\geq s\geq 0$. Completely positive $V(t,s)$ correspond to a Markovian evolution with CP-divisible $\Lambda(t)$ \cite{RHP,Wolf}. Therefore, if $V(t,s)$ is not completely positive, then the evolution is non-Markovian. Further characterization is provided by a hierarchy of $k$-positive $V(t,s)$ that give rise to $k$-divisible dynamical maps \cite{Sabrina}. The smallest value of $k=1$ recovers P-divisible $\Lambda(t)$, associated with weakly non-Markovian evolution. For invertible dynamical maps, the propagator $V(t,s)=\Lambda(t)\Lambda^{-1}(s)$ is always well defined, and the CP-divisibility property is equivalent to the positivity of the decoherence rates $J_\alpha(t)$ entering the time-local generator in the Gorini-Kossakowski-Sudarshan-Lindblad (GKSL) form \cite{GKS,L}
\begin{equation}
\mathcal{L}(t)[\rho]=-i[H(t),\rho]+\sum_\alpha J_\alpha(t)\Bigg(V_\alpha(t)\rho V_\alpha^\dagger(t)-\frac 12\{V_\alpha^\dagger(t)V_\alpha(t),\rho\}\Bigg).
\end{equation}

One family of dynamical maps whose divisibility properties have been well characterized are Nathanson and Ruskai's generalized Pauli channels \cite{Ruskai}. They describe qudit depolarization that depends on the choice of axis, which is defined via mutually unbiased bases \cite{Wootters}. Due to their symmetry properties, these channels  find important applications in quantum information theory, quantum communication, and open quantum systems. They have been extensively analyzed in terms of their geometrical properties \cite{invertibility_measure,GGPC_volume}, classical capacity \cite{Holevo_capacity,Engineering_capacity}, parameter estimation \cite{Shin2}, and channel fidelity \cite{norms,fidelity}. It has been shown that the generalized Pauli channels define time-evolution with numerous negative decoherence rates \cite{ICQC,CCMS}. Moreover, classical mixtures of non-invertible generalized Pauli channels are known to produce a Markovian semigroup \cite{CCMK,Noninvertibility}.

In this paper, we consider qudit evolution that goes beyond the class established by the generalized Pauli channels. We start by introducing two families of quantum channels constructed from symmetric measurements \cite{SIC-MUB}. Just like the generalized Pauli channels, these maps are highly symmetric and therefore have the potential to be widely applicable. Their clear advantages over the generalized Pauli channels are that (i) the channels from symmetric measurements can be constructed in any dimension $d$ and (ii) there exist multiple families of quantum channels in any given dimension. This is a direct result of the fact that there exist at least four families of symmetric measurements in any dimension $d>2$. Up to this point, symmetric measurements have been gaining significant attention in entanglement detection, where they are used in the formulation of improved separability criteria \cite{SICMUB_App,SICMUB_App2,Lai,SICMUB_App4}, optimal local entanglement detection \cite{Alber}, and construction of bound entanglement witnesses \cite{SICMUB_Pmaps}. They have also been implemented to characterize entropic uncertainty relations \cite{SICMUB_entropic,SICMUB_entropic2}, average coherence \cite{SICMUB_App3}, steerability \cite{Alber2}, the Brukner-Zeilinger invariants \cite{SICMUB_BZ}, as well as optimal state estimation \cite{SICMUB_design}. Recently, more restrictive conditions for the existence of optimal symmetric measurements have been presented \cite{Alber3}.

In the following sections, we first recall the definition and properties of symmetric measurements. Next, we use these measurement operators to define two classes of quantum channels via mixtures of entanglement breaking channels. We derive conditions for complete positivity and entanglement breaking, and show how they relate to conical 2-designs. For the associated dynamical maps, we formulate sufficient P and CP-divisibility conditions, which allows us to provide non-Markovianity criteria that depend on the coefficients that characterize symmetric measurements. Special attention is given to D-divisibility, which interpolates between P- and CP-divisibility \cite{D-div}. Finally, we present illustrative examples by deriving non-Markovianity criteria for the channels constructed from mutually unbiased bases, mutually unbiased measurements, and rank-2 projective measurement operators. It turns out that completely copositive maps allow for detecting more P-divisible dynamical maps than sufficient P-divisibility conditions.

\section{Symmetric measurements}

Every quantum measurement can be represented via a positive, operator-valued measure (POVM), which consist in positive operators $E_\alpha$ summing up to the identity operator $\mathbb{I}_d$, where $d=\mathrm{dim}\mathcal{H}$. The probability associated with a POVM element $E_\alpha$ for a density operator $\rho$ is then $p_\alpha=\Tr(E_\alpha\rho)$. Due to having a simple mathematical structure, measurements with symmetry properties are of significant interest.

One way of imposing symmetry conditions onto the measurement operators has recently been presented in ref. \cite{SIC-MUB}, which introduces the notion of symmetric measurements. Namely, a symmetric measurement, or $(N,M)$-POVM, is a collection of $N$ POVMs $\{E_{\alpha,k};\,k=1,\ldots,M\}$, $\alpha=1,\ldots,N$, provided that the following symmetry constraints are satisfied,
\begin{equation}\label{M}
\begin{split}
\Tr (E_{\alpha,k})&=\frac dM,\\
\Tr (E_{\alpha,k}^2)&=x,\\
\Tr (E_{\alpha,k}E_{\alpha,\ell})&=\frac{d-Mx}{M(M-1)}\equiv y,\qquad \ell\neq k,\\
\Tr (E_{\alpha,k}E_{\beta,\ell})&=\frac{d}{M^2}\equiv z,\qquad \beta\neq\alpha,
\end{split}
\end{equation}
where $x$ is its characterizing parameter from the range
\begin{equation}\label{x}
\frac{d}{M^2}<x\leq\min\left\{\frac{d^2}{M^2},\frac{d}{M}\right\}.
\end{equation}
Projective measurements are recovered for $x=d^2/M^2$ and $M\geq d$.
An $(N,M)$-POVM is informationally complete if and only if $N=\frac{d^2-1}{M-1}$. For $d=2$, this leads to only two valid choices of $M$ and $N$, which correspond to the general SIC POVMs $(N=1,M=4)$ \cite{Gour} and the mutually unbiased measurements $(N=3,M=2)$ \cite{Kalev}. However, for any higher dimension $d$, there are at least four classes of informationally complete $(N,M)$-POVMs:
\begin{enumerate}[label=(\roman*)]
\item $M=d^2$ and $N=1$ (general SIC POVMs),
\item $M=d$ and $N=d+1$ (MUMs),
\item $M=2$ and $N=d^2-1$,
\item $M=d+2$ and $N=d-1$.
\end{enumerate}

Importantly, informationally complete $(N,M)$-POVMs can be constructed in any finite dimension using orthonormal Hermitian operator bases $\{G_0=\mathbb{I}_d/\sqrt{d},G_{\alpha,k};\,\alpha=1,\ldots,N,\,k=1,\ldots,M-1\}$ such that $\Tr G_{\alpha,k}=0$. The measurement operators are given by
\begin{equation}\label{E}
E_{\alpha,k}=\frac 1M \mathbb{I}_d+tH_{\alpha,k},
\end{equation}
where
\begin{equation}\label{H}
H_{\alpha,k}=\left\{\begin{aligned}
&G_\alpha-\sqrt{M}(\sqrt{M}+1)G_{\alpha,k},\quad k=1,\ldots,M-1,\\
&(\sqrt{M}+1)G_\alpha,\qquad k=M,
\end{aligned}\right.
\end{equation}
and $G_\alpha=\sum_{k=1}^{M-1}G_{\alpha,k}$. The new parameter $t$ is directly related to $x$ through the equation
\begin{equation}\label{xt}
x=\frac{d}{M^2}+t^2(M-1)(\sqrt{M}+1)^2,
\end{equation}
and its admissible range depends on the choice of $G_{\alpha,k}$.

\section{Quantum channels from measurement operators}

Among quantum channels, there exists a class that is simulable via classical channels. It is realized as follows: the sender performs measurements via a POVM $\{F_k\}$ on a state $\rho$, sends the outcomes $k$ via a classical channel, and the receiver prepares the corresponding states $R_k$. The resulting quantum channel has the Holevo form \cite{EBC}
\begin{equation}
\Phi[\rho]=\sum_kR_k\Tr(F_k\rho).
\end{equation}
From any given $(N,M)$-POVM, we can in this way construct a family of completely positive, trace-preserving maps
\begin{equation}
\Phi_\alpha[X]=\frac{M}{d}\sum_{k=1}^ME_{\alpha,k}\Tr(XE_{\alpha,k}),
\end{equation}
where $\overline{E}_{\alpha,k}$ denotes the complex conjugation of ${E}_{\alpha,k}$. Obviously, these maps have the Holevo form and therefore are entanglement breaking. Now, $\Phi_\alpha$ obey the following composition relations for any $\beta\neq\alpha$,
\begin{equation}
\Phi_\alpha\Phi_\alpha=
\frac Md \Big[(x-y)\Phi_\alpha+My\Phi_0\Big],\qquad
\Phi_\alpha\Phi_\beta=\Phi_0,
\end{equation}
where $\Phi_0[X]=\mathbb{I}_d\Tr(X)/d$ is the maximally depolarizing channel. Additionally, for an informationally complete $(N,M)$-POVM,
\begin{equation}\label{suma}
\sum_{\alpha=1}^N\Phi_\alpha=\frac Md (\kappa_+d\Phi_0+\kappa_-\oper)
\end{equation}
with the coefficients
\begin{equation}
\kappa_+=\frac{M}{d}[y+(N-1)z]=\frac{d^3-xM^2}{dM(M-1)},\qquad \kappa_-=x-y=\frac{xM^2-d}{M(M-1)}
\end{equation}
that appear in the formula for conical 2-designs \cite{SICMUB_design} (see Appendix A). Observe that $\Phi_\alpha$ admits the eigenvalue equations
\begin{equation}
\Phi_\alpha[U_{\beta,\ell}]=
\frac Md (x-y)\delta_{\alpha\beta}U_{\beta,\ell},
\end{equation}
where the eigenvectors
\begin{equation}
U_{\alpha,k}=\sum_{\ell=1}^M\omega^{k\ell}E_{\alpha,\ell},\qquad \omega=\exp\left(\frac{2\pi i}{M}\right),\qquad k=1,\ldots,M-1,
\end{equation}
are orthogonal operators satisfying the trace relation
\begin{equation}
\Tr(U_{\alpha,k}U_{\beta,\ell}^\dagger)=M(x-y)\delta_{\alpha\beta}\delta_{k\ell}.
\end{equation}
Notably, $U_{\alpha,k}$ are unitary if $E_{\alpha,k}$ are projective measurements.

Now, let us introduce the quantum channel that is a classical mixture of entanglement breaking channels $\Phi_\alpha$ and the identity; namely,
\begin{equation}\label{lambda}
\Lambda=\frac{dp_0-1}{d-1}\oper+\frac{d}{d-1}\sum_{\alpha=1}^Np_\alpha\Phi_\alpha,
\end{equation}
where $p_0\geq 1/d$, $p_\alpha\geq 0$, and $\sum_{\alpha=0}^Np_\alpha=1$.
If $E_{\alpha,k}$ are mutually unbiased measurements, this construction reproduces the generalization of Pauli channels from ref. \cite{MUM_GPC}. For mutually unbiased bases, one recovers the generalized Pauli channels \cite{Ruskai,Ohno}.
Observe that $\Lambda$ is a unital map ($\Lambda[\mathbb{I}_d]=\mathbb{I}_d$) and commutes with all $\Phi_\alpha$. Its eigenvalue equations
\begin{equation}
\Lambda[U_{\alpha,k}]=\lambda_\alpha U_{\alpha,k}
\end{equation}
show that this map has real, $(M-1)$-times degenerated eigenvalues
\begin{equation}
\lambda_\alpha=\frac{1}{d-1}[dp_0+M(x-y)p_\alpha-1].
\end{equation}
From the inverse relation,
\begin{equation}
\begin{split}
p_0&=\frac{1}{dM\kappa_+}\Big[(d-1)\sum_{\alpha=1}^N\lambda_\alpha+N-M(x-y)\Big],\\
p_\alpha&=\frac{d-1}{M(x-y)}\left[\lambda_\alpha-\frac{d\sum_{\beta=1}^N\lambda_\beta-M(x-y)}{dM\kappa_+}\right],
\end{split}
\end{equation}
one easily derives the sufficient complete positivity conditions
\begin{equation}\label{CPT}
\frac{\kappa_-}{d}\leq\frac{1}{M}\sum_{\alpha=1}^N\lambda_\alpha\leq
\frac{\kappa_-}{d}+\kappa_+\min_\alpha\lambda_\alpha.
\end{equation}
Notably, this implies that $\lambda_\alpha\geq 0$.


An alternative construction can be performed by taking a classical mixture
\begin{equation}\label{lambda2}
\widetilde{\Lambda}=\frac{dq_0-1}{d-1}\oper+\frac{d}{d-1}\sum_{\alpha=1}^Nq_\alpha\Psi_\alpha
\end{equation}
over another family of entanglement breaking channels
\begin{equation}
\Psi_\alpha=\frac{1}{M-1}(M\Phi_0-\Phi_\alpha).
\end{equation}
Note that $\Psi_\alpha$ indeed have the Holevo form as
\begin{equation}
\Psi_\alpha[X]=\sum_{k=1}^MR_{\alpha,k}\Tr(M_{\alpha,k}X),
\end{equation}
where $R_{\alpha,k}=\frac Md E_{\alpha,k}$ are density operators and $M_{\alpha,k}=\frac{1}{M-1}(\mathbb{I}_d-E_{\alpha,k})$ form a POVM. 
It is easy to show that
\begin{equation}
\Psi_\alpha[U_{\beta,\ell}]=
-\frac{M(x-y)}{d(M-1)}\delta_{\alpha\beta}U_{\beta,\ell},
\end{equation}
and hence the eigenvalues of $\widetilde{\Lambda}$ read
\begin{equation}
\widetilde{\Lambda}[U_{\alpha,k}]=\widetilde{\lambda}_\alpha U_{\alpha,k},
\qquad \widetilde{\lambda}_\alpha=\frac{1}{d-1}\left[dq_0-1-\frac{M(x-y)}{M-1}q_\alpha\right].
\end{equation}
The probability distribution can be expressed in terms of the channel eigenvalues in the following way,
\begin{align}
q_0&=\frac{(d-1)(M-1)\sum_{\alpha=1}^N\widetilde{\lambda}_\alpha+d^2-1+M(x-y)}{d(d^2-1)+M(x-y)},\\
q_\alpha&=\frac{(d-1)(M-1)}{M(x-y)}\left[\frac{d(M-1)\sum_{\beta=1}^N\widetilde{\lambda}_\beta+M(x-y)}
{d(d^2-1)+M(x-y)}-\widetilde{\lambda}_\alpha\right],
\end{align}
and hence $\widetilde{\Lambda}$ is completely positive if
\begin{equation}\label{CPT2}
-\frac{M(x-y)}{d(M-1)}\leq\sum_{\alpha=1}^N\widetilde{\lambda}_\alpha\leq N,\qquad
\max_\alpha\widetilde{\lambda}_\alpha\leq\frac{d(M-1)\sum_{\beta=1}^N\widetilde{\lambda}_\beta+M(x-y)}
{d(d^2-1)+M(x-y)}.
\end{equation}
Contrary to the construction from $\Phi_\alpha$, the channel $\widetilde{\Lambda}$ allows for negative eigenvalues. For the sake of clarity, let us point out that $\widetilde{\Lambda}$ can be represented in the same form as $\Lambda$, the main difference being that the coefficients multiplying $\oper$ and $\Phi_\alpha$ are no longer a probability distribution.


Even though our construction uses entanglement breaking channels, the resulting map is in general no longer entanglement breaking. This is the case due to mixing with the identity operator. However, it is still possible to derive sufficient conditions for breaking quantum entanglement. 

\begin{Proposition}
A quantum channel constructed from symmetric measurements is entanglement breaking if
\begin{equation}
\left\{\begin{split}
&\lambda_\alpha\geq 0,\\
&\sum_{\alpha=1}^{d+1}\lambda_\alpha\leq \frac{M(x-y)}{d}
\end{split}\right.
\qquad\mathrm{or}\qquad
\left\{\begin{split}
&\lambda_\alpha\leq 0,\\
&\sum_{\alpha=1}^{d+1}\left|\lambda_\alpha\right|\leq\frac{M(x-y)}{d(M-1)}.
\end{split}\right.
\end{equation}
\end{Proposition}

\begin{proof}
Observe that the channels from eqs. (\ref{lambda}) and (\ref{lambda2}) are equivalently rewritten as
\begin{equation}
\Lambda=\left(1-\frac{d}{M(x-y)}\sum_{\alpha=1}^N\lambda_\alpha\right)\Phi_0
+\frac{d}{M(x-y)}\sum_{\alpha=1}^N\lambda_\alpha\Phi_\alpha,
\end{equation}
\begin{equation}
\widetilde{\Lambda}=\left(1+\frac{d(M-1)}{M(x-y)}\sum_{\alpha=1}^N\lambda_\alpha\right)\Phi_0
-\frac{d(M-1)}{M(x-y)}\sum_{\alpha=1}^N\lambda_\alpha\Psi_\alpha.
\end{equation}
A classical mixture of entanglement breaking channels is also entanglement breaking, and hence it is enough that the coefficients before $\Phi_0$, $\Phi_\alpha$, and $\Psi_\alpha$ are positive.
\end{proof}

\section{Non-Markovianity criteria}

A time-dependent description is realized via time-parameterized families of quantum channels $\{\Lambda(t);\,t\geq 0,\,\Lambda(0)=\oper\}$ known as dynamical maps. For the channel in eq. (\ref{lambda}), all the time dependence is encoded in the probability distribution $p_\alpha(t)$, so that the corresponding dynamical map reads
\begin{equation}\label{lambdat}
\Lambda(t)=\frac{dp_0(t)-1}{d-1}\oper+\frac{d}{d-1}\sum_{\alpha=1}^Np_\alpha(t)\Phi_\alpha.
\end{equation}
Quantum dynamics is often specified by master equations, which are the evolution equations for $\Lambda(t)$. Assume that $\Lambda(t)$ is the solution of the  master equation $\dot{\Lambda}(t)=\mathcal{L}(t)\Lambda(t)$ with the time-local  generator
\begin{equation}\label{La}
\mathcal{L}(t)[\rho]=\Phi(t)[\rho]
-\frac 12 \{\Phi^{\#}(t),\rho\},
\end{equation}
where $\Phi(t)=\sum_{\alpha=1}^N\gamma_\alpha(t)\Phi_\alpha$ and $\Phi^{\#}(t)$ denotes a map dual to $\Phi(t)$. Equivalently, $\mathcal{L}(t)$ can be written as
\begin{equation}\label{gen}
\mathcal{L}(t)=\sum_{\alpha=1}^N\gamma_\alpha(t)(\Phi_\alpha-\oper).
\end{equation}
The generator satisfies the following eigenvalue equations,
\begin{equation}
\mathcal{L}(t)[U_{\alpha,k}]=\xi_\alpha(t)U_{\alpha,k},\qquad
\xi_\alpha(t)=\frac Md (x-y)\gamma_\alpha(t)
-\sum_{\beta=1}^N\gamma_\beta(t),
\end{equation}
and hence $\Lambda(t)$ has time-dependent eigenvalues
\begin{equation}
\lambda_\alpha(t)=\exp\left[\int_0^t\xi_\alpha(\tau)\der\tau\right].
\end{equation}
Now, if $\mathcal{L}(t)$ is constructed from a completely positive map $\Phi(t)$, then the dynamical map $\Lambda(t)$ is CP-divisible. This corresponds to the Markovian evolution. Note that, depending on the choice of symmetric measurements, the form in eq. (\ref{gen}) may not de diagonal. Therefore, in general, $\gamma_\alpha(t)\geq 0$ are not the necessary and sufficient  CP-divisibility conditions.

\begin{Proposition}
If the time-local generator $\mathcal{L}(t)$ possesses the coefficients $\gamma_\alpha(t)$ such that
\begin{equation}\label{CPCOND}
\gamma_\alpha(t)\geq 0\qquad{\rm{or}}\qquad \sum_{\beta=1}^N\gamma_\beta(t)
-(N-\kappa_+)\gamma_\alpha(t)\geq 0,
\end{equation}
then $\Lambda(t)$ is a CP-divisible dynamical map.
\end{Proposition}

For the proof, see Appendix B. Note that both conditions are equivalent for the generators of the generalized Pauli channels, where $\kappa_+=1$ and $N>2$.

Among non-Markovian evolutions, special attention is given to quantum dynamics with P-divisible dynamical maps. Recall that an invertible dynamical map is P-divisible if and only if
\begin{equation}
\frac{\der}{\der t}\|\Lambda(t)[X]\|_1\leq 0
\end{equation}
for all operators $X$, where $\|X\|_1=\Tr\sqrt{X^\dagger X}$ is the trace norm of $X$ \cite{CKR}. In general, this condition is hard to verify. However, if one takes $X=U_\alpha$, then the necessary conditions for P-divisibility follow,
\begin{equation}
\xi_\alpha(t)\leq 0\qquad\Longleftrightarrow\qquad 
-\frac Md (x-y)\gamma_\alpha(t)+\sum_{\beta=1}^N\gamma_\beta(t)\geq 0.
\end{equation}
To derive sufficient conditions, we use the fact that if $\Lambda(t)$ is generated via $\mathcal{L}(t)$ defined via eq. (\ref{La}) with a positive map $\Phi(t)$, then $\Lambda(t)$ is P-divisible. As the positive map constructed from symmetric measurements, we take \cite{SICMUB_Pmaps}
\begin{equation}\label{Pmap}
\Phi=(b-N+2L)\mu_0\Phi_0+\sum_{\alpha=L+1}^N\mu_\alpha\Phi_\alpha
-\sum_{\alpha=1}^L\mu_\alpha\Phi_\alpha
\end{equation}
with $b=M(d-1)(x-y)/d$, as well as
\begin{enumerate}[label=(\roman*)]
\item $\mu_0\geq\max_{\beta=1,\ldots,L}\mu_\beta$ and $\mu_\alpha\geq\max_{\beta=1,\ldots,L}\mu_\beta$ for $\alpha=L+1,\ldots,N$ if $L\geq (N-b)/2$; or,
\item $\mu_0\leq\min_{\alpha=L+1,\ldots,N}\mu_\alpha$ and $\mu_\alpha\geq\max_{\beta=1,\ldots,L}\mu_\beta$ for $\alpha=L+1,\ldots,N$ if $L< (N-b)/2$.
\end{enumerate}

\begin{Proposition}
If at any time $t\geq 0$ the generator $\mathcal{L}(t)$ has
\begin{equation}\label{Lcond}
\frac{N-b}{2}<L<\frac{N-b+M\kappa_+}{2}
\end{equation}
negative coefficients $\gamma_\beta(t)$, $\beta=1,\ldots,L$, and $N-L$ positive coefficients $\gamma_\alpha(t)$, $\alpha=L+1,\ldots,N$, that satisfy
\begin{equation}\label{PCOND}
\bigforall_{\alpha=L+1,\ldots,N}\quad
\gamma_\alpha(t)\geq\frac{M\kappa_++(b-N+2L)}{M\kappa_+-(b-N+2L)}
\max_{\beta=1,\ldots,L}|\gamma_\beta(t)|,
\end{equation}
then $\Lambda(t)$ is a P-divisible dynamical map.

If at any time $t\geq 0$ the generator $\mathcal{L}(t)$ has $L\leq (N-b)/2$ 
negative coefficients $\gamma_\beta(t)$, $\beta=1,\ldots,L$, and $N-L$ positive coefficients $\gamma_\alpha(t)$, $\alpha=L+1,\ldots,N$, such that
\begin{equation}\label{PCOND2}
\bigforall_{\alpha=L+1,\ldots,N}\quad
\gamma_\alpha(t)\geq\max_{\beta=1,\ldots,L}|\gamma_\beta(t)|,
\end{equation}
then $\Lambda(t)$ is a P-divisible dynamical map.
\end{Proposition}

The proofs are given in Appendix C. Notably, this approach does not allow for $L=N$.

Recently, a classification of dynamical maps has been introduced that interpolates between P and CP-divisibility. It is based on the notion of decomposable maps, which are positive maps of the form $\Phi(t)=\mathcal{A}(t)+T\mathcal{B}(t)$ with completely positive maps $\mathcal{A}(t)$, $\mathcal{B}(t)$ and transposition $T$. Generators constructed from such maps give rise to D-divisible (decomposable divisible) dynamical maps $\Lambda(t)$ \cite{D-div}. However, to determine whether a map is decomposable is a demanding task. For this reason, we consider a special subclass consisting of D-divisible $\Lambda(t)$ with corresponding $\Phi(t)$ that are completely copositive (completely positive after transposition).

\begin{Proposition}
A dynamical map $\Lambda(t)$ is D-divisible if the coefficients in its generator $\mathcal{L}(t)$ satisfy
\begin{equation}\label{coCPCOND}
\sum_{\alpha=1}^N\gamma_\alpha(t)
\sum_{k=1}^ME_{\alpha,k}\otimes E_{\alpha,k}\geq 0.
\end{equation}
\end{Proposition}

The condition for complete copositivity of $\Phi(t)=\sum_{\alpha=1}^N\gamma_\alpha\Phi_\alpha$ follows directly from the positivity of the Choi matrix
\begin{equation}
C[T\Phi(t)]=(\oper\otimes T\Phi(t))[dP_+]=\sum_{\alpha=1}^N\gamma_\alpha(t)
\sum_{k=1}^M\overline{E}_{\alpha,k}\otimes\overline{E}_{\alpha,k},
\end{equation}
where $P_+=(1/d)\sum_{m,n=1}^d|m\>\<n|\otimes|m\>\<n|$ denotes the maximally entangled state. For comparison, $\Phi(t)$ is completely positive if and only if
\begin{equation}
C[\Phi(t)]=\sum_{\alpha=1}^N\gamma_\alpha(t)
\sum_{k=1}^M\overline{E}_{\alpha,k}\otimes E_{\alpha,k}.
\end{equation}
Hence, if $E_{\alpha,k}=\overline{E}_{\alpha,k}$ for $\gamma_\alpha(t)\neq 0$, then $\Phi(t)$ is both completely positive and completely copositive. Examples include incomplete sets of mutually unbiased bases and symmetric operators constructed from the Gell-Mann matrices, which is analyzed in more details in the next section.

\section{Examples}

\subsection{Mutually unbiased bases}

For $N=d+1$, $M=d$, the symmetric measurements reduce to complete sets of mutually unbiased bases. As mentioned above, they are used to construct the generalized Pauli channels \cite{Ruskai}. The sufficient complete positivity conditions in eqs. (\ref{CPT}) and (\ref{CPT2}) reduce to
\begin{equation}
1\leq \sum_{\beta=1}^{d+1}\lambda_\beta\leq 1+d\min_\alpha\lambda_\alpha
\end{equation}
and
\begin{equation}
\frac{1}{d-1}\left(-1+d^2\max\{0,\lambda_\alpha\}\right)\leq \sum_{\beta=1}^{d+1}\lambda_\beta\leq d+1,
\end{equation}
respectively. Therefore, even in this simple case, they are not equivalent to  the well-known generalized Fujiwara-Algoet conditions \cite{Fujiwara}
\begin{equation}
-\frac{1}{d-1}\leq \sum_{\beta=1}^{d+1}\lambda_\beta\leq 1+d\min_\alpha\lambda_\alpha,
\end{equation}
which are sufficient and necessary for complete positivity of the generalized Pauli channels. However, the CP-divisibility conditions $\gamma_\alpha(t)\geq 0$ for $\Lambda(t)$ following from eq. (\ref{CPCOND}) become sufficient and necessary, where $\gamma_\alpha(t)$ are the decoherence rates of the generator $\mathcal{L}(t)$. This can be easily seen by taking the sum over $\beta\neq\alpha$ in the second inequality, so that
\begin{equation}
(1-\kappa_+)\sum_{\beta=1}^N\gamma_\beta(t)+(N-2)\gamma_\alpha(t)\geq 0,
\end{equation}
and recalling that for MUBs $\kappa_+=1$.

The D-divisibility is in general hard to verify. For a special case with $\gamma_1(t)=\gamma(t)$ and $\gamma_2(t)=\ldots=\gamma_{d+1}(t)=\widetilde{\gamma}(t)$, we use the results for conical designs in eq. (\ref{P2}) to simplify eq. (\ref{coCPCOND}) into
\begin{equation}
\widetilde{\gamma}(t)(\mathbb{I}_{d^2}+\mathbb{F}_d)+
[\gamma(t)-\widetilde{\gamma}(t)]\sum_{k=1}^d|k\>\<k|\otimes|k\>\<k|\geq 0.
\end{equation}
Hence, $\Lambda(t)$ is D-divisible if
\begin{equation}
\widetilde{\gamma}(t)\geq 0,\qquad \gamma(t)\geq-\widetilde{\gamma}(t),
\end{equation}
which is equivalent to the sufficient P-divisibility conditions.

More results can be obtained after fixing the dimension. In $d=3$, the sufficient D-divisibility conditions read as follows,
\begin{equation}\label{ineq}
\sum_{\beta=1}^4\sum_{\alpha\neq\beta}\gamma_\alpha(t)\gamma_\beta(t)\geq 0.
\end{equation}
Note that, contrary to the P- and CP-divisibility constraints, this inequality is non-linear in $\gamma_\alpha(t)$.

\begin{Remark}
Consider a qutrit evolution that admits two positive and two negative coefficients
\begin{equation}
\gamma_1(t)=\gamma_2(t)=\gamma(t)\geq 0,\qquad \gamma_3(t)=\gamma_4(t)=-1.
\end{equation}
Then, eq. (\ref{ineq}) holds for any ordering of the MUBs provided that
\begin{equation}
\gamma(t)\geq\frac{\sqrt{3}+1}{\sqrt{3}-1}.
\end{equation}
This condition is stronger than the sufficient P-divisibility conditions from Proposition 3, which only detect $\gamma(t)\geq 5$.
\end{Remark}

\subsection{Mutually unbiased measurements}

Consider the orthonormal Hermitian basis in $d=3$, consisting in $G_0=\mathbb{I}_3/\sqrt{3}$ and the Gell-Mann matrices
\begin{align*}
G_{1,1}=\frac{1}{\sqrt{2}}\begin{pmatrix}
1 & 0 & 0 \\
0 & -1 & 0 \\
0 & 0 & 0
\end{pmatrix},&\qquad
G_{1,2}=\frac{1}{\sqrt{6}}\begin{pmatrix}
1 & 0 & 0 \\
0 & 1 & 0 \\
0 & 0 & -2
\end{pmatrix},\\
G_{2,1}=\frac{1}{\sqrt{2}}\begin{pmatrix}
0 & 1 & 0 \\
1 & 0 & 0 \\
0 & 0 & 0
\end{pmatrix},&\qquad
G_{2,2}=\frac{1}{\sqrt{2}}\begin{pmatrix}
0 & 0 & 1 \\
0 & 0 & 0 \\
1 & 0 & 0
\end{pmatrix},\\
G_{3,1}=\frac{1}{\sqrt{2}}\begin{pmatrix}
0 & 0 & 0 \\
0 & 0 & 1 \\
0 & 1 & 0
\end{pmatrix},&\qquad
G_{3,2}=\frac{1}{\sqrt{2}}\begin{pmatrix}
0 & 0 & 0 \\
0 & 0 & -i \\
0 & i & 0
\end{pmatrix},\\
G_{4,1}=\frac{1}{\sqrt{2}}\begin{pmatrix}
0 & -i & 0 \\
i & 0 & 0 \\
0 & 0 & 0
\end{pmatrix},&\qquad
G_{4,2}=\frac{1}{\sqrt{2}}\begin{pmatrix}
0 & 0 & -i \\
0 & 0 & 0 \\
i & 0 & 0
\end{pmatrix}.
\end{align*}
This basis can be used to construct non-optimal ($x<x_{\max}=1$) mutually unbiased measurements with $x=5/9$ according to eqs. (\ref{E}) and (\ref{H}) with $M=d$ and $N=d+1$. The corresponding dynamical map $\Lambda(t)$ is CP-divisible if and only if
\begin{equation}\label{CP3}
\left\{
\begin{split}
&\gamma_2(t)+\gamma_3(t)+\gamma_4(t)\geq 0,\\
&2\gamma_3(t)-\gamma_4(t)+5\gamma_2(t)\geq 0,\\
&2\gamma_3(t)-\gamma_2(t)+5\gamma_4(t)\geq 0,\\
&3\gamma_1(t)+\gamma_2(t)-2\gamma_3(t)+\gamma_4(t)\geq 0,\\
&6\gamma_1(t)+2\gamma_2(t)+5\gamma_3(t)+2\gamma_4(t)\pm \sqrt{2[\gamma_2(t)+\gamma_4(t)]^2+\gamma_3^2(t)}\geq 0,
\end{split}
\right.
\end{equation}
which means that, contrary to the MUB case, the sufficient conditions
\begin{equation}
\gamma_\alpha(t)\geq 0\qquad\vee\qquad 9\sum_{\beta=1}^{4}
\gamma_\beta(t)-25\gamma_\alpha(t)\geq 0
\end{equation}
from eq. (\ref{CPCOND}) do not detect all Markovian evolutions. 
Interestingly, if $\gamma_3(t)=[\gamma_2(t)+\gamma_4(t)]/2$, then the necessary and sufficient CP-divisibility conditions simply reduce to $\gamma_\alpha(t)\geq 0$.

Now, $\Lambda(t)$ is D-divisible if
\begin{equation}\label{coCP3}
\left\{
\begin{split}
&\gamma_3(t)=\frac{\gamma_2(t)+\gamma_4(t)}{2},\\
&\gamma_2(t)+\gamma_4(t)\geq 0,\\
&2\gamma_1(t)+\gamma_2(t)+\gamma_4(t)\geq \sqrt{2}|\gamma_2(t)-\gamma_4(t)|.
\end{split}
\right.
\end{equation}
Therefore, $\Lambda(t)$ constructed from a completely copositive $\Phi(t)$ is also CP-divisible if
\begin{equation}
\gamma_3(t)=\frac{\gamma_2(t)+\gamma_4(t)}{2},\qquad \gamma_2(t)\geq 0,\qquad
\gamma_4(t)\geq 0,\qquad 2\gamma_1(t)+\gamma_2(t)+\gamma_4(t)\geq \sqrt{2}|\gamma_2(t)-\gamma_4(t)|.
\end{equation}
This is a more general result than one would initially assume from the formulas for the Choi matrices $C[\Phi(t)]$, $C[T\Phi(t)]$ and the properties of the Gell-Mann matrices $\overline{G}_{\alpha,k}=G_{\alpha,k}$ for $\alpha=1,2$.

\subsection{Optimal $(15,2)$-POVM}

Finally, we analyze the ququart channels that arise from the following rank-2 projectors,
\begin{equation}
E_{\alpha,\pm}=\frac 12 \mathbb{I}_4\pm G_\alpha,\qquad \alpha=1,\ldots,15,
\end{equation}
where $\{G_\alpha\}=\{\sigma_k\otimes\sigma_\ell\}$ consists of tensor products of the normalized Pauli matrices $\sigma_k$. Now, using $E_{\alpha,\pm}$, we construct the completely positive maps
\begin{equation}
\begin{split}
\Phi_\alpha[X]&=\frac 12 \Big[E_{\alpha,+}\Tr(E_{\alpha,+}X)+E_{\alpha,-}\Tr(E_{\alpha,-}X)\Big]
=\frac 14 \mathbb{I}_4\Tr X+G_\alpha\Tr(G_\alpha X),
\end{split}
\end{equation}
and then
\begin{equation}
\Phi(t)=\sum_{\alpha=1}^{15}\gamma_\alpha(t)\Phi_\alpha,\qquad \mathrm{where}\qquad\Phi(t)[G_\alpha]=\gamma_\alpha(t)G_\alpha.
\end{equation}
Hence, the corresponding time-local generator $\mathcal{L}(t)$ is determined by the coefficients $\gamma_\alpha(t)$, which are in this case simply the eigenvalues of $\Phi(t)$ to the eigenvectors $G_\alpha$. It is straightforward to find the diagonal form of the generator
\begin{equation}
\mathcal{L}(t)[X]=\frac 14 \sum_{\alpha=1}^{15}J_\alpha(t)(4G_\alpha XG_\alpha-X)
\end{equation}
with time-dependent decoherence rates $J_\alpha(t)$. From Proposition 3, the sufficient P-divisibility conditions allow for $L=12$ negative $\gamma_\beta(t)$ and $N-L=3$ positive $\gamma_\alpha(t)$ such that $\gamma_\alpha(t)\leq 13\max_{\beta=1,\ldots,12}|\gamma_\beta(t)|$. In particular, if we take $\gamma_\beta(t)=-1$ and $\gamma_\alpha(t)=13$, then the time-local generator admits two negative decoherence rates. However, there does not seem to be a straightforward correlation between $L$ and the number of negative rates. For example, $L=6$ allows for six negative $J_\alpha(t)$, whereas $L=7$ can result in a CP-divisible generator with only six non-zero $J_\alpha(t)$.

\section{Conclusions}

In this paper, we expanded on the concept of generalized Pauli channels by constructing quantum maps using symmetric measurements rather than mutually unbiased bases. This allowed us to find a wide class of qudit dynamics, whose properties -- like complete positivity and degrees of divisibility -- are still relatively easy to control. It is important to note that our results are general; i.e., they do not depend on the explicit choice of symmetric measurements. The more parameters we fix before performing derivations, the more restrictive conditions we obtain, which is illustrated in the low-dimensional examples of qubit, qutrit, and ququart evolution. Interestingly, our results suggest that completely copositive maps can be more efficient in detecting weakly Markovian dynamics than the known positive maps.

In further works, it would be beneficial to establish less restrictive positivity and complete positivity criteria of dynamical maps. This would by extension allow for more reliable detection methods of both Markovian and non-Markovian quantum evolution. Another interesting aspect would be to find a procedure of constructing the diagonal form of the corresponding time-local generator. Partial results for special classes of symmetric measurements, other than mutually unbiased bases, would also be desirable, Finally, the notion of D-divisibile dynamical maps, which lie between P-divisible and CP-divisible maps, requires further analysis. One thing to consider is whether there are any advantages of quantum maps that are both complete positive and complete copositive, which can be viewed as a dynamical analogue of PPT entangled states.

\section*{Acknowledgements}

This research was funded in whole or in part by the National Science Centre, Poland, Grant number 2021/43/D/ST2/00102. For the purpose of Open Access, the author has applied a CC-BY public copyright licence to any Author Accepted Manuscript (AAM) version arising from this submission.

\bibliography{C:/Users/cyndaquilka/OneDrive/Fizyka/bibliography}
\bibliographystyle{C:/Users/cyndaquilka/OneDrive/Fizyka/beztytulow2}

\appendix

\section{Conical 2-designs}

The concept of conical designs was introduced by Appleby and Graydon \cite{Graydon,Graydon2} as a generalization of complex projective designs to positive operators. From definition, a conical 2-design is a family of positive operators $A_j$ such that $\sum_jA_j\otimes A_j$ commutes with $U\otimes U$ for any unitary operator $U$.This property is equivalent to 
\begin{equation}
\sum_jA_j\otimes A_j=
\kappa_+\mathbb{I}_d\otimes\mathbb{I}_d+\kappa_-\mathbb{F}_d
\end{equation}
for $\kappa_+\geq\kappa_->0$ with $\mathbb{F}_d=\sum_{m,n=1}^d|m\>\<n|\otimes|n\>\<m|$ being the flip operator \cite{Graydon}. 
Recently, it has been shown that informationally complete symmetric measurements are conical 2-designs  \cite{SICMUB_design}. The proof uses elaborate results derived from mutually unbiased measurements \cite{Wang}. Here, we show that informationally complete $(N,M)$-POVMs being conical 2-designs follows directly from the properties of $\Phi_\alpha$.

\begin{Lemma}
For the family of completely positive, trace-preserving maps
\begin{equation}
\Phi_\alpha[X]=\frac Md \sum_{k=1}^ME_{\alpha,k}\Tr(E_{\alpha,k}X)
\end{equation}
constructed from an informationally complete $(N,M)$-POVM, it holds that
\begin{equation}\label{kap}
\sum_{\alpha=1}^N\Phi_\alpha=\frac Md (\kappa_+d\Phi_0+\kappa_-\oper)
\end{equation}
with the coefficients
\begin{equation}\label{kappas}
\kappa_+=\frac Md [y+(N-1)z]=\frac{d^3-xM^2}{dM(M-1)},\qquad \kappa_-=x-y=\frac{xM^2-d}{M(M-1)}.
\end{equation}
\end{Lemma}

\begin{proof}
We start by multiplying $\Phi_\alpha[X]$ by $E_{\beta,\ell}$ and taking the trace. Using the properties of $E_{\beta,\ell}$, we arrive at
\begin{equation}
\begin{split}
\frac dM \sum_{\alpha=1}^N\Tr(\Phi_\alpha[X]E_{\beta,\ell})&=\sum_{\alpha=1}^N
\sum_{k=1}^M\Tr(E_{\alpha,k}E_{\beta,\ell})\Tr(E_{\alpha,k}X)
=x\Tr(E_{\beta,\ell}X)+\sum_{k\neq\ell}y\Tr(E_{\beta,k}X)
+\sum_{\alpha\neq\beta}\sum_{k=1}^Mz\Tr(E_{\alpha,k}X)\\
&=(x-y)\Tr(E_{\beta,\ell}X)+[y+(N-1)z]\Tr X
=\kappa_-\Tr(E_{\beta,\ell}X)+\frac dM \kappa_+\Tr X.
\end{split}
\end{equation}
Equivalently, this can be rewritten into the condition
\begin{equation}\label{trace}
\Tr\left\{\left(\frac dM\sum_{\alpha=1}^N\Phi_\alpha[X]-
\Big[\kappa_-X+\frac dM \kappa_+\mathbb{I}_d\Tr(X)\Big]
\right)E_{\beta,\ell}\right\}=0
\end{equation}
that holds for any Hermitian $X$ and $E_{\beta,\ell}$. Therefore, the operator in the round brackets is equal to zero.
\end{proof}

Using the above results, it is straightforward to prove that $(N,M)$-POVMs are conical 2-designs.

\begin{Theorem}
Any informationally complete $(N,M)$-POVM is a conical 2-design with
\begin{equation}\label{P2}
\sum_{\alpha=1}^N\sum_{k=1}^ME_{\alpha,k}\otimes E_{\alpha,k}
=\kappa_+\mathbb{I}_d\otimes\mathbb{I}_d + \kappa_- \mathbb{F}_d,
\end{equation}
where $\kappa_{\pm}$ are given by eq. (\ref{kappas}).
\end{Theorem}

\begin{proof}
From the Choi-Jamio{\l}kowski isomorphism, we construct the Choi map
\begin{equation}
C(\Phi T)=\frac dM \sum_{\alpha=1}^N(\oper_d\otimes \Phi_\alpha T)[dP_+],
\end{equation}
where $P_+=(1/d)\sum_{m,n=1}^d|m\>\<n|\otimes|m\>\<n|$ is the maximally entangled state and $T$ denotes the transposition. Direct calculations lead to
\begin{equation}
C(\Phi T)=\frac dM \sum_{\alpha=1}^N\sum_{m,n=1}^d|m\>\<n|\otimes\Phi_\alpha[|n\>\<m|]
=\sum_{\alpha=1}^N\sum_{k=1}^M\sum_{m,n=1}^d
|m\>\<m|E_{\alpha,k}|n\>\<n|\otimes E_{\alpha,k}=
\sum_{\alpha=1}^N\sum_{k=1}^ME_{\alpha,k}\otimes E_{\alpha,k}.
\end{equation}
On the other hand, using the results of Lemma 1, we have
\begin{equation}
C(\Phi T)=\left(\oper_d\otimes\Big[\kappa_+d\Phi_0+\kappa_-T\Big]\right)[dP_+]
=\sum_{m,n=1}^d|m\>\<n|\otimes\Big[\kappa_+\mathbb{I}_d\delta_{mn}+\kappa_-|n\>\<m|\Big]
=\kappa_+\mathbb{I}_d\otimes\mathbb{I}_d+\kappa_-\mathbb{F}_d.
\end{equation}
Finally, observe that $\kappa_+\geq\kappa_->0$, which follows from the upper and lower bounds on the parameter $x$ in eq. (\ref{x}).
\end{proof}

\section{CP-divisibility conditions}

In Section 3, we have found two families of quantum channels constructed from symmetric measurements. Therefore, in deriving sufficient conditions for complete positivity of $\Phi(t)$, we need to consider two seperate cases. If $\Phi(t)=\sum_{\alpha=1}^N\gamma_\alpha(t)\Phi_\alpha$, then it is completely positive if $\gamma_\alpha(t)\geq 0$. On the other hand, if we construct $\Phi(t)$ as a combination of $\Psi_\alpha$, then
\begin{equation}
\Psi(t)=\sum_{\alpha=1}^N\zeta_\alpha(t)\Psi_\alpha=\frac{1}{(M-1)\kappa_+}
\sum_{\alpha=1}^N\Phi_\alpha\left[\sum_{\beta=1}^N\zeta_\beta(t)-\kappa_+\zeta_\alpha(t)\right]
-\frac{M(x-y)}{(M-1)d\kappa_+}\sum_{\beta=1}^N\zeta_\beta(t)\oper.
\end{equation}
A gauge transformation allows us to eliminate the term proportional to $\oper$, so that
\begin{equation}\label{zeta}
\Psi^\prime(t)=\sum_{\alpha=1}^N\gamma_\alpha(t)\Phi_\alpha,\qquad \gamma_\alpha(t)=\frac{1}{(M-1)\kappa_+}
\left[\sum_{\beta=1}^N\zeta_\beta(t)-\kappa_+\zeta_\alpha(t)\right].
\end{equation}
The inverse relation reads
\begin{equation}\label{zeta2}
\zeta_\alpha(t)=\frac{M-1}{N-\kappa_+}\left[\sum_{\beta=1}^N\gamma_\beta(t)
-(N-\kappa_+)\gamma_\alpha(t)\right].
\end{equation}
The complete positivity conditions follow from $\zeta_\alpha(t)\geq 0$, which is equivalent to the second inequality in Proposition 3.

\section{P-divisibility conditions}

First, using formula (\ref{kap}), we rewrite the positive map in eq. (\ref{Pmap}) into
\begin{equation}
\Phi(t)=\eta_0(t)\oper+\sum_{\alpha=L+1}^N\eta_\alpha(t)\Phi_\alpha
-\sum_{\alpha=1}^L\eta_\alpha(t)\Phi_\alpha,
\end{equation}
where for simplicity we introduce the coefficients
\begin{equation}
\begin{split}
\eta_0(t)&=-\frac{\mu_0(t)\kappa_-}{M\kappa_+}(b-N+2L),\\
\eta_\alpha(t)&=\mu_\alpha(t)-\frac{\mu_0(t)}{M\kappa_+}(b-N+2L),\quad \alpha=1,\ldots,L,\\
\eta_\alpha(t)&=\mu_\alpha(t)+\frac{\mu_0(t)}{M\kappa_+}(b-N+2L),\quad \alpha=L+1,\ldots,N.
\end{split}
\end{equation}
Now, observe that $\Phi_\alpha$ are self-dual in the sense that $\Phi_\alpha^{\#}=\Phi_\alpha$, where a dual map is defined via the trace relation $\Tr(X\Phi[Y])=\Tr(\Phi^{\#}[X]Y)$. Therefore, $\Phi(t)$ is self-dual, as well, which allows us to calculate
\begin{equation}
\Phi^{\#}(t)[\mathbb{I}]=\left[\mu_0(t)\left(1+\frac{N}{M\kappa_+}(b-N+2L)\right)
+\sum_{\alpha=L+1}^N\mu_\alpha(t)-\sum_{\alpha=1}^L\mu_\alpha(t)\right]\mathbb{I}.
\end{equation}
After comparing the two formulas for the generator in eqs. (\ref{gen}) and (\ref{La}), we identify the generator coefficients
\begin{equation}
\begin{split}
\gamma_\alpha(t)&=-\mu_\alpha(t)+\frac{\mu_0(t)}{M\kappa_+}(b-N+2L),\qquad \alpha=1,\ldots,L,\\
\gamma_\alpha(t)&=\mu_\alpha(t)+\frac{\mu_0(t)}{M\kappa_+}(b-N+2L),\qquad \alpha=L+1,\ldots,N.
\end{split}
\end{equation}
Recall that $\Lambda(t)$ is P-divisible if $\Phi(t)$ that defines the generator $\mathcal{L}(t)$ is positive. However, the positivity of $\Phi(t)$ depends on the range of $L$, which means we have to consider two separate cases.

If $L\geq (N-b)/2$, then $\Phi(t)$ is positive if
\begin{equation}\label{CC}
\mu_0(t)\geq\max_{\beta=1,\ldots,L}\mu_\beta(t),\qquad
\mu_\alpha(t)\geq\max_{\beta=1,\ldots,L}\mu_\beta(t),\qquad\alpha=L+1,\ldots,N.
\end{equation}
Interestingly, $\mu_0(t)$ is an extra parameter that is not fully determined by $\gamma_\alpha(t)$. Therefore, to obtain the least restrictive conditions, we choose the smallest value
\begin{equation}\label{mu0}
\mu_0(t)=\max_{\beta=1,\ldots,L}\mu_\beta(t)
=\frac{M\kappa_+}{M\kappa_+-b+N-2L}\max_{\beta=1,\ldots,L}|\gamma_\beta(t)|.
\end{equation}
Simple calculations show that the second condition in eq. (\ref{CC}) is now equivalent to
\begin{equation}
\bigforall_{\alpha=L+1,\ldots,N}\quad\gamma_\alpha(t)\geq
\frac{M\kappa_++(b-N+2L)}{M\kappa_+-(b-N+2L)}\max_{\beta=1,\ldots,L}|\gamma_\beta(t)|.
\end{equation}
The upper limit for the number $L$ of negative coefficients follows from the requirement that $\mu_0(t)\geq 0$, which reduces to
\begin{equation}
M\kappa_+-b+N-2L>0.
\end{equation}
Note that $M\kappa_+-b=N-M\kappa_-<N$, and therefore one always has $L<N$.

On the other hand, if $L<(N-b)/2$, then $\Phi(t)$ is positive if
\begin{equation}\label{CC2}
\mu_0(t)\leq\min_{\alpha=L+1,\ldots,N}\mu_\alpha(t),\qquad
\mu_\alpha(t)\geq\max_{\beta=1,\ldots,L}\mu_\beta(t),\qquad\alpha=L+1,\ldots,N.
\end{equation}
Again, we can eliminate $\mu_0(t)$ by taking its minimal value, which this time is $\mu_0(t)=0$. Hence, the above inequalities reduce to
\begin{equation}
\bigforall_{\alpha=L+1,\ldots,N}\quad\gamma_\alpha(t)\geq
\max_{\beta=1,\ldots,L}|\gamma_\beta(t)|.
\end{equation}


\end{document}